\makeatletter
\let\@twosidetrue\@twosidefalse
\let\@mparswitchtrue\@mparswitchfalse
\makeatother

\documentclass[format=acmsmall, review=false]{acmart}

\usepackage{acm-ec-18}
\usepackage[utf8]{inputenc}
\usepackage[english]{babel}

\usepackage{booktabs} 
\usepackage[ruled]{algorithm2e} 

\SetAlFnt{\small}
\SetAlCapFnt{\small}
\SetAlCapNameFnt{\small}
\SetAlCapHSkip{0pt}
\IncMargin{-\parindent}

\usepackage{complexity}
\usepackage{graphicx}
\usepackage{float}
\usepackage{xcolor,xspace}
\usepackage{paralist}
\usepackage{refcount}
\usepackage{hyperref}
\usepackage{amsthm, algpseudocode}
\usepackage[textsize=tiny,textwidth=1.2cm,shadow]{todonotes}

\newtheorem{remark}[theorem]{Remark}
\usetikzlibrary{arrows,decorations.markings,patterns,calc, positioning, backgrounds}
\tikzstyle{vertex} = [circle, draw=black, scale=0.7]
\tikzstyle{edgelabel} = [circle, fill=white, scale=0.8]

\begin{document}
\title[Cake cutting with unequal shares]{The complexity of cake cutting with unequal shares}  
\author{\'{A}gnes Cseh}
\affiliation{
Institute of Economics, Hungarian Academy of Sciences
}
\author{Tam\'{a}s Fleiner}
\affiliation{Department of Computer Science and Information Theory, Budapest University of Technology and Economics
}

\begin{abstract}
An unceasing problem of our prevailing society is the fair division of goods. The problem of proportional cake cutting focuses on dividing a heterogeneous and divisible resource, the cake, among $n$ players who value pieces according to their own measure function. The goal is to assign each player a not necessarily connected part of the cake that the player evaluates at least as much as her proportional share.

In this paper, we investigate the problem of proportional division with unequal shares, where each player is entitled to receive a predetermined portion of the cake. Our main contribution is threefold. First we present a protocol for integer demands that delivers a proportional solution in fewer queries than all known algorithms. Then we show that our protocol is asymptotically the fastest possible by giving a matching lower bound. Finally, we turn to irrational demands and solve the proportional cake cutting problem by reducing it to the same problem with integer demands only. All results remain valid in a highly general cake cutting model, which can be of independent interest.
\end{abstract}
\maketitle

\section{Introduction} 

In cake cutting problems, the cake symbolizes a heterogeneous and divisible resource that shall be distributed among $n$ players. Each player has her own measure function, which determines the value of any part of the cake for her. The aim of proportional cake cutting is to allocate each player a piece that is worth at least as much as her proportional share, evaluated with her measure function~\cite{Ste48}. The measure functions are not known to the protocol. 

The efficiency of a fair division protocol can be measured by the number of queries. In the standard Robertson-Webb model~\cite{RW98}, two kinds of queries are allowed. The first one is the \emph{cut} query, in which a player is asked to mark the cake at a distance from a given starting point so that the piece between these two is worth a given value to her. The second one is the \emph{eval} query, in which a player is asked to evaluate a given piece according to her measure~function. 

If shares are meant to be \emph{equal} for all players, then the proportional share is defined as $\frac{1}{n}$ of the whole cake. In the \emph{unequal shares} version of the problem (also called cake cutting with entitlements), proportional share is defined as a player-specific demand, summing up to the value of the cake over all players. The aim of this paper is to determine the query complexity of proportional cake cutting in the case of unequal shares. \citet{RW98} write in their seminal book ``Nothing approaching general theory of optimal number of cuts for unequal shares division has been given to date. This problem may prove to be very difficult.'' We now settle the issue for the number of queries, the standard measure of efficiency instead of the number of physical cuts.

\subsection{Related work}

\textbf{Equal shares} Possibly the most famous cake cutting protocol belongs to the class of Divide and Conquer algorithms. Cut and Choose is a 2-player equal-shares protocol that guarantees proportional shares. It already appeared in the Old Testament, where Abraham divided Canaan to two equally valuable parts and his brother Lot chose the one he valued more for himself. The first $n$-player variant of this algorithm is attributed to Banach and Knaster~\cite{Ste48} and it requires $\mathcal{O}\left(n^2\right)$ queries. Other methods include the continuous (but discretizable) Dubins-Spanier protocol~\cite{DS61} and the Even-Paz protocol~\cite{EP84}. The latter show that their method requires $\mathcal{O}\left(n \log n\right)$ queries at most. The complexity of proportional cake cutting in higher dimensions has been studied in several papers~\cite{BBS09, Bec87, BJK08, Hil83, IH09, SNHA17}, in which cuts are tailored to fit the shape of the cake. 

\textbf{Unequal shares} The problem of proportional cake cutting with unequal shares is first mentioned by \citet{Ste48}. Motivated by dividing a leftover cake, \citet{RW98} define the problem formally and offer a range of solutions for two players. More precisely, they list cloning players, using Ramsey partitions~\cite{MRW92} and most importantly, the Cut Near-Halves protocol~\cite{RW98}. The last method computes a fair solution for 2 players with integer demands $d_1$ and $d_2$ in $2\lceil\log_2(d_1+d_2)\rceil$ queries. Robertson and Webb also show how any 2-player protocol can be generalized to $n$ players in a recursive manner. 

\textbf{Irrational demands} The case of irrational demands in the unequal shares case is interesting from the theoretical point of view, but beyond this, solving it might be necessary, because other protocols might generate instances with irrational demands. For example, in the maximum-efficient envy-free allocation problem with two players and piecewise linear measure functions, any optimal solution must be specified using irrational numbers, as \citet{CLPP11} show. \citet{Bar96} studies the case of cutting the cake in an irrational ratio between $n$ players and presents an algorithm that constructs a proportional division. \citet{SZ99} solve the same problem with the objective of minimizing the number of resulting pieces. Their protocol is simpler than that of \citet{Bar96}. 

\textbf{Lower bounds} The drive towards establishing lower bounds on the complexity of cake cutting protocols is coeval to the cake cutting literature itself~\cite{Ste48}. \citet{EP84} conjectured that their protocol is the best possible, while Robertson and Webb explicitly write that ``they would place their money against finding a substantial improvement on the $n \log_2{n}$ bound'' for proportional cake cutting with equal shares. After approximately 20 years of no breakthrough in the topic, \citet{MBK03} showed that any protocol must make
$\Omega(n \log_2{n})$ comparisons -- but this was no bound on the number of queries. Essentially simultaneously, \citet{SW07} came up with the lower bound $\Omega(n \log_2{n})$ on the number of queries for the case where contiguous pieces are allocated to each player. Not much later, this condition was dropped by \citet{EP11} who completed the query complexity analysis of proportional cake cutting with equal shares by presenting a lower bound of $\Omega(n \log_2{n})$. \citet{BJK11} study the minimum number of actual cuts in the case of unequal shares and prove that $n-1$ cuts might not suffice -- in other words, they show that there is no proportional allocation with contiguous pieces. However, no lower bound on the number of queries has been known in the case of unequal shares.

\textbf{Generalizations in higher dimensions} There are two sets of multiple-dimensional generalizations of the proportional cake cutting problem. The first group focuses on the existence of a proportional division, without any constructive proof. The existence can be shown easily using Lyapunov’s theorem, as stated by \citet{DS61} as Corollary 1.1. \citet{Ber92} investigate the existence of envy-free divisions. \citet{Dal01} considers the case of equal shares and defines a dual optimization problem that allows to compute a proportional solution by minimizing convex functions over a finite dimensional simplex. Complexity issues are not discussed in these papers, in fact, queries are not even mentioned in them. 

The second group of multiple-dimensional generalizations considers problems where certain geometric parameters are imposed on the cake and the pieces, see \citet{BBS09,Bec87,BJK08,Hil83,IH09,SNHA07}. Also, some of these have special extra requirements on the output, such as contiguousness or envy-freeness. These works demonstrate the interest in various problems in multi-dimensional cake cutting, for which we define a very general framework.

\subsection{Our contribution}

We provide formal definitions in Section~\ref{sec:prel} and present the query analysis of the fastest known protocol for the $n$-player proportional cake cutting problem with total demand $D$ in Section~\ref{sec:known}. Then, in Section~\ref{sec:protocol} we focus on our protocol for the problem, which is our main contribution in this paper. The idea is that we recursively render the players in two batches so that these batches can simulate two players who aim to cut the cake into two approximately equal halves. Our protocol requires only $2\left(n-1\right) \cdot \lceil\log_2{D}\rceil$ queries. Other known protocols reach  $D \cdot \lceil\log_2{D}\rceil$ and $n(n-1) \cdot \lceil\log_2{D}\rceil$, thus ours is the fastest procedure that derives a proportional division for the $n$-player cake cutting problem with unequal shares. Moreover, our protocol also works on a highly general cake (introduced in Section~\ref{sec:gen}), extending the traditional notion of the cake to any finite dimension.

We complement our positive result by showing a lower bound of $\Omega\left(n\cdot \log_2{D}\right)$ on the query complexity of the problem in Section~\ref{sec:lower_bound}. Our proof generalizes, but does not rely on, the lower bound proof given by \citet{EP11} for the problem of proportional division with equal shares. Moreover, our lower bound remains valid in the generalized cake cutting and query model, allowing a considerably more powerful notion of a query even on the usual, $[0,1]$ interval cake.

In Section~\ref{sec:irrat} we turn to irrational demands and solve the proportional cake cutting problem by reducing it to the same problem with integer demands only. By doing so, we provide a novel and simple approach to the problem. Moreover, our method works in the generalized query model as well.

\section{Preliminaries}
\label{sec:prel}

We begin with formally defining our input. Our setting includes a set of players of cardinality $n$, denoted by $\{P_1, P_2, \dots, P_n\}$, and a heterogeneous and divisible good, which we refer to as the cake and project to the unit interval $[0,1]$. Each player $P_i$ has a non-negative, absolutely continuous \emph{measure function} $\mu_i$ that is defined on Lebesgue-measurable sets. We remark that absolute continuity implies that every zero-measure set has value 0 according to $\mu_i$ as well. In particular, $\mu_i((a,b)) = \mu_i([a,b])$ for any interval $[a,b] \subseteq [0,1]$. Besides measure functions, each player $P_i$ has a \emph{demand} $d_i \in \mathbb{Z}^+$, representing that $P_i$ is entitled to receive $d_i/\sum\limits_{j=1}^n d_j \in ]0,1[$ 
part of the whole cake. The value of the whole cake is identical for all players, in particular it is the sum of all demands:
\begin{equation*}
\forall 1 \leq i \leq n \quad \mu_i([0,1])= D = \sum\limits_{j=1}^n d_j.
\end{equation*}
We remark that an equivalent formulation is also used sometimes, where the demands are rational numbers that sum up to 1, the value of the full cake. Such an input can be transformed into the above form simply by multiplying all demands by the least common denominator of all demands. As opposed to this, if demands are allowed to be irrational numbers, then no ratio-preserving transformation might be able to transform them to integers. That is why the case of irrational demands is treated separately.

The cake $[0,1]$ will be partitioned into subintervals in the form $[x,y), 0 \leq x \leq y \leq 1$. A finite union of such subintervals forms a \emph{piece} $X_i$ allocated to player~$P_i$. We would like to stress that a piece is not necessarily connected.

\begin{definition}
A set $\lbrace X_i \rbrace_{1 \leq i \leq n}$ of pieces is a \emph{division} of the cake $[0,1]$ if $\bigcup\limits_{1 \leq i \leq n}X_i =[0,1]$ and $X_i \cap X_j = \emptyset$ for all $i \neq j$. We call division $\lbrace X_i \rbrace_{1 \leq i \leq n}$ \emph{proportional} if $\mu_i(X_i)\geq d_i$ for all $1 \leq i \leq n$.
\end{definition}

In words, proportionality means that each player receives a piece with which her demand is satisfied. We do not consider Pareto optimality or alternative fairness notions such as envy-freeness in this paper.

We now turn to defining the measure of efficiency in cake cutting. We assume that $1 \leq i \leq n$, $x,y \in [0,1]$ and $0\leq \alpha \leq 1$. Oddly enough, the Robertson-Webb query model was not formalized explicitly by Robertson and Webb first, but by \citet{SW07}, who attribute it to the earlier two. In their query model, a protocol can ask agents the following two types of queries.

\begin{itemize}
\item \textit{Cut query} $(P_i,\alpha)$ returns the leftmost point $x$ so that $\mu_i([0,x]) = \alpha$. In this operation $x$ becomes a so-called \textit{cut point}.
\item \textit{Eval query} $(P_i,x)$ returns $\mu_i([0,x])$. Here $x$ must be a cut point.
\end{itemize}

Notice that this definition implies that choosing sides, sorting marks or calculating any other parameter than the value of a piece are not counted as queries and thus they do not influence the efficiency of a protocol. 

\begin{definition}
The \emph{number of queries} in a protocol is the number of eval and cut queries until termination. We denote the number of queries for a $n$-player algorithm with total demand $D$ by $T(n, D)$.
\end{definition}

The query definition of Woeginger and Sgall is the strictest of the type Robertson-Webb. We now outline three options to extend the notion of a query, all of which have been used in earlier papers~\cite{EP11,EP84,RW98,SW07} and are also referred to as Robertson-Webb queries.

\begin{enumerate}
\item \textbf{The query definition of Edmonds and Pruhs.} There is a slightly different and stronger formalization of the core idea, given by \citet{EP11} and also used by \citet{Pro13, Pro15}. The crucial difference is that they allow both cut and eval queries to start from an arbitrary point in the cake. 

\begin{itemize}
\item \textit{Cut query} $(P_i, x, \alpha)$ returns the leftmost point $y$ so that $\mu_i([x,y]) =\alpha$ or an error message if no such $y$ exists.
\item \textit{Eval query} $(P_i,x, y)$ returns $\mu_i([x,y])$.
\end{itemize}

These queries can be simulated as trivial concatenations of the queries defined by Woeginger and Sgall. To pin down the starting point $x$ of a cut query $(P_i, x, \alpha)$ we introduce the cut point $x$ with the help of a dummy player's Lebesgue-measure, ask $P_i$ to evaluate the piece $[0,x]$ and then we cut query with value $\alpha' = \alpha + \mu_i([0,x])$. Similarly, to generate an eval query $(P_i,x, y)$ one only needs to artificially generate the two cut points $x$ and $y$ and then ask two eval queries of the Woeginger-Sgall model, $(P_i,x)$ and $(P_i,y)$. We remark that such a concatenation of Woeginger-Sgall queries reveals more information than the single query in the model of Edmonds and Pruhs.

\item  \textbf{Proportional cut query.} The term \textit{proportional cut query} stands for generalized cut queries of the sort ``$P_i$ cuts the piece $[x,y]$ in ratio $a:b$'', where $a,b$ are integers. As Woeginger and Sgall also note it, two eval queries and one cut query with ratio $\alpha = \frac{a}{a+b} \cdot \mu_i([x,y])$ are sufficient to execute such an operation if $x,y$ are cut points, otherwise five queries suffice. Notice that the eval queries are only used by $P_i$ when she calculates $\alpha$, and their output does not need to be revealed to any other player or even to the protocol.

\item \textbf{Reindexing.} When working with recursive algorithms it is especially useful to be able to reindex a piece $[x,y]$ so that it represents the interval $[0,1]$ for~$P_i$. Any further cut and eval query on $[x,y]$ can also be substituted by at most five queries on the whole cake. Similarly as above, there is no need to reveal the result of the necessary eval queries addressed to a player.
\end{enumerate}

These workarounds ensure that protocols require asymptotically the same number of queries in both model formulations, even if reindexing and proportional queries are allowed. We opted for utilizing all three extensions of the Woeginger-Sgall query model in our upper bound proofs, because the least restrictive model allows the clearest proofs. Regarding our lower bound proof, it holds even if we allow a highly general query model including all of the above extensions, which we define in Section~\ref{sec:gen_query}.

\section{Known protocols}
\label{sec:known}

To provide a base for comparison, we sketch the known protocols for proportional cake cutting with unequal shares and bound their query complexity.

The most naive approach to the case of unequal shares is the cloning technique, where each player $P_i$ with demand $d_i$ is substituted by $d_i$ players with unit demands. In this way a $D$-player equal shares cake cutting problem is generated, which can be solved in $\mathcal{O}(D \log_2{D})$ queries~\cite{EP84}.

As \citet{RW98} point out, any 2-player protocol can be generalized to an $n$-player protocol. They list two 2-player protocols, Cut Near-Halves and the Ramsey Partition Algorithm~\citep{MRW92} and also remark that for 2 players, Cut Near-Halves is always at least as efficient as Ramsey Partition Algorithm. Therefore, we restrict ourselves to analyzing the complexity of the generalized Cut Near-Halves protocol. 

Cut Near-Halves is a simple procedure, in which the cake of value $D$ is repeatedly cut in approximately half by players $P_1$ and $P_2$ with demands $d_1 \leq d_2$ as follows. $P_1$ cuts the cake into two near-halves, more precisely, in ratio $\lfloor \frac{D}{2} \rfloor : \lceil \frac{D}{2} \rceil$. Then, $P_2$ picks a piece that she values at least as much as~$P_1$. This piece is awarded to $P_2$ and her claim is reduced accordingly, by the respective near-half value of the cake. In the next round, the same is repeated on the remaining part of the cake, and so on, until $d_1$ or $d_2$ is reduced to zero. Notice that the cutter is always the player with the lesser current demand, and thus this role might be swapped from round to round.

The recursive $n$-player protocol of Robertson and Webb runs as follows. We assume that $k-1 < n$ players, $P_1, P_2, \dots, P_{k-1}$, have already divided the whole cake of value $D = d_1 + d_2 + \ldots + d_n$. The next player $P_k$ then challenges each of the first $k-1$ players separately to redistribute the piece already assigned to them. In these rounds, $P_k$ claims $\frac{d_k}{d_1 + d_2 + \ldots + d_{k-1}}$ part of each piece. 
This generates $k-1$ rounds of the Cut Near-Halves protocol, each with 2 players. Notice that this protocol tends to assign a highly fractured piece of cake to every player.

The following theorem summarizes the results known about the complexity of the 2-player and $n$-player versions of the Cut Near-Halves protocol. 

\begin{theorem}[\citet{RW98}]
\label{th:unequal_old}
The 2-player Cut Near-Halves protocol with demands $d_1, d_2$ requires $T(2) = 2\lceil\log_2(d_1+d_2)\rceil$ queries at most. The recursive $n$-player version is finite.
\end{theorem}

Here we give an estimate for the number of queries of the recursive protocol.

\begin{theorem}
The number of queries in the recursive $n$-player Cut Near-Halves protocol is at most
\begin{equation*}
T(n,D)=\sum\limits_{i=1}^{n-1} \biggl[2i\cdot\Bigl\lceil\log_2\Bigl(\sum\limits_{j=1}^{i+1} d_j\Bigr)\Big\rceil\biggr]  \leq n(n-1) \cdot \lceil\log{D}\rceil.
\end{equation*}
\label{tet:unequal_ours}
\end{theorem}

\begin{proof}
The first round consists of players $P_1$ and $P_2$ sharing the cake using $2\lceil\log_2{(d_1+d_2)}\rceil$ queries. The second round then has two 2-player runs, each of them requiring $2\lceil\log_2{(d_1+d_2+d_3)}\rceil$ queries. In general, the $i$th round terminates after $i \cdot 2\lceil\log_2{\Bigl(\sum\limits_{j=1}^{i+1} d_j\Bigr)}\rceil$ queries at most. The number of rounds is $n-1$. Now we add up the total number of queries.

\begin{equation*}
\begin{split}
\sum\limits_{i=1}^{n-1} \biggl[2i\cdot\Bigl\lceil\log_2\Bigl(\sum\limits_{j=1}^{i+1} d_j\Bigr)\Big\rceil\biggr]\leq
\sum\limits_{i=1}^{n-1} \biggl[2i\cdot\Bigl\lceil\log_2\Bigl(\sum\limits_{j=1}^{n} d_j\Bigr)\Big\rceil\biggr]=
\sum\limits_{i=1}^{n-1} \biggl[2i\cdot\Bigl\lceil\log_2{D}\Big\rceil\biggr]=
n(n-1) \cdot \lceil\log{D}\rceil
\end{split}
\end{equation*}
\end{proof}

The following example proves that the calculated bound can indeed be reached asymptotically in instances with an arbitrary number of players.
\begin{example}
The estimation for the query number is asymptotically sharp if $\Bigl\lceil\log_2\Bigl(\sum\limits_{j=1}^{i+1} d_j\Bigr)\Big\rceil = \Bigl\lceil\log_2{D}\Big\rceil$ holds for at least a fixed portion of all $1 \leq i \leq n-1$, say, for the third of them. This is easy to reach if $n$ is a sufficiently large power of 2 and all but one players have demand $1$, while there is another player with demand~$2$. Notice that this holds for every order for the agents. If one sticks to a decreasing order of demand when indexing the players, then not only asymptotic, but also strict equality can be achieved by setting $d_1$ much larger than all other demands.
\end{example}

\section{Our protocol}
\label{sec:protocol}

In this section, we present a simple and elegant protocol that beats all three above mentioned protocols in query number. Our main idea is that we recursively render the players in two batches so that these batches can simulate two players who aim to cut the cake into two approximately equal halves. For now we work with the standard cake and query model defined in Section~\ref{sec:prel}. Later, in Section~\ref{sec:gen_prot} we will show how our protocol can be extended to a more general cake. We remind the reader that cutting near-halves means to cut in ratio $\lfloor \frac{D}{2} \rfloor : \lceil \frac{D}{2} \rceil$.

To ease the notation we assume that the players are indexed so that when they mark the near-half of the cake, the marks appear in an increasing order from 1 to~$n$. In the subsequent rounds, we reindex the players to keep this property intact. Based on these marks, we choose ``the middle player'', this being the player whose demand reaches the near-half of the cake when summing up the demands in the order of marks. This player cuts the cake and each player is ordered to the piece her mark falls to. The middle player is cloned if necessary so that she can play on both pieces. The protocol is then repeated on both generated subinstances, with adjusted demands. In the subproblem, the players' demands are according to the ratios listed in the pseudocode. 

\begin{center}
\hspace*{0.2cm}\fbox{
\hspace*{2.5mm}\hspace*{-1\fboxsep}
\parbox{0.9\linewidth}{
\begin{minipage}[tb]{0.95\linewidth}
\begin{description}
    \item[\textbf{Proportional division with unequal shares}]
\end{description}
Each player marks the near-half of the cake $X$.\\ Sort the players according to their marks.\\
Calculate the smallest index $j$ such that $\lfloor \frac{D}{2} \rfloor\leq\sum_{i=1}^{j}{d_{i}}=:a$.\\
Cut the cake in two along $P_j$'s mark.\\
Define two instances of the same problem and solve them recursively. 
\begin{enumerate}
\item Players $P_1, P_2, \dots, P_j$ share piece $X_1$ on the left. Demands are set to $d_1, d_2 \dots, d_{j-1}, d_j-a+\lfloor \frac{D}{2}\rfloor$, while measure functions are set to $\mu_i \cdot \lfloor\frac{D}{2}\rfloor / \mu_i(X_1)$, for all $1\leq i\leq j$.
\item Players  $P_j, P_{j+1}, \ldots, P_n$ share piece $X_2=X \setminus X_1$ on the right. Demands are set to $a-\lfloor \frac{D}{2} \rfloor, d_{j+1}, d_{j+2}, \dots, d_n$, while measure functions are set to $\mu_i \cdot \lceil\frac{D}{2}\rceil / \mu_i(X_2)$, for all $j\leq i\leq n$.
\end{enumerate}
\end{minipage}
}}
\end{center} 
\vspace*{0.1cm}

\begin{example}
\label{ex:unequal}
We present our protocol on an example with $n=3$. Every step of the protocol is depicted in Figure~\ref{fig:unequal}. Let $d_1=1, d_2=3, d_3=1$. Since $D=5$ is odd, all players mark the near-half of the cake in ratio 2:3. The cake is then cut at $P_2$'s mark, since $d_1 < \lfloor\frac{D}{2}\rfloor$, but $d_1 + d_2 \geq \lfloor\frac{D}{2}\rfloor$. The first subinstance will consist of players $P_1$ and $P_2$, both with demand 1, whereas the second subinstance will have the second copy of player $P_2$ alongside $P_3$ with demands 2 and 1, respectively. In the first instance, both players mark half of the cake and the one who marked it closer to 0 will receive the leftmost piece, while the other player is allocated the remaining piece. The players in the second instance mark the cake in ratio $1:2$. Suppose that the player demanding more marks it closer to 0. The leftmost piece is then allocated to her and the same two players share the remaining piece in ratio $1:1$. The player with the mark on the left will be allocated the piece on the left, while the other players takes the remainder of the piece.
These rounds require $3+2+2+2 = 9$ proportional cut queries and no eval query.
\end{example}

\begin{figure}[t]
\centering
\begin{tikzpicture}[every edge/.style={shorten <=1pt, shorten >=1pt}]

  \draw[thick] (0,0)  node [below] {0} -- (10,0) node [below] {1};
  \foreach \x/\xtext in {4.2/$P_3$,3.4/$P_2$,2.6/$P_1$}
    \draw[thick] (\x,1pt) -- (\x,-1pt) node[below] {\xtext};
\node[draw] at (11,0) {2:3};
 	\draw[->,rounded corners] (2.6,-0.5) -| node[pos=0.75,fill=white,inner sep=2pt]{1} ++(-0.5,-0.4) |- (1.7,-1.2);
	\draw[->,rounded corners] (3.35,-0.5) -|node[pos=0.75,fill=white,inner sep=2pt]{1} ++(-0.5,-0.4) |-  (2.4,-1.2);
	\draw[->,rounded corners] (4.2,-0.5) -| node[pos=0.75,fill=white,inner sep=2pt]{1} ++(0.5,-0.4) |- (5.2,-1.2);
	\draw[->,rounded corners] (3.45,-0.5) -| node[pos=0.75,fill=white,inner sep=2pt]{2} ++(0.5,-0.4) |-  (4.4,-1.2);

  \draw[thick] (0,-1.5)  node [below] {0} -- (3.3,-1.5) node [below] {1};
  \draw (3.5,-1.5) node [below] {0} -- (10,-1.5) node [below] {1};
  \foreach \x/\xtext in {1.7/$P_1$,2.4/$P_2$,4.6/$P_2$,6.4/$P_3$}
  \draw[thick] (\x,-1.5pt-1.5cm) -- (\x,-2+1.5pt-1.5cm) node[below] {\xtext};
\node[draw] at (11,-1.5) {1:2};
\node[draw] at (-1,-1.5) {1:1};

 	\draw[->,rounded corners] (4.6,-2) -| node[pos=0.75,fill=white,inner sep=2pt]{1} ++(0.5,-0.4) |- (5.6,-2.7);
    \draw[->,rounded corners] (6.4,-2) -| node[pos=0.75,fill=white,inner sep=2pt]{1} ++(0.5,-0.4) |- (7.4,-2.7);

	\draw[very thick, green] (0,-3)--  node [below] {$P_1$} (1.65,-3);
    \draw[very thick, red] (1.75,-3)--  node [below] {$P_2$} (3.35,-3);
  \draw[very thick, red] (3.45,-3)--  node [below] {$P_2$} (4.55,-3);
  \draw (4.65,-3) node [below] {0} -- (10,-3) node [below] {1};
  \foreach \x/\xtext in {8.6/$P_2$,6.7/$P_3$}
  \draw[thick] (\x,-1.5pt-3cm) -- (\x,-2+1.5pt-3cm) node[below] {\xtext};
\node[draw] at (11,-3) {1:1};

	\draw[very thick, green] (0,-4.5)--  node [below] {$P_1$} (1.65,-4.5);
    \draw[very thick, red] (1.75,-4.5)--  node [below] {$P_2$} (3.35,-4.5);
  \draw[very thick, red] (3.45,-4.5)--  node [below] {$P_2$} (4.55,-4.5);
    \draw[very thick, blue] (4.65,-4.5)--  node [below] {$P_3$} (6.65,-4.5);
    \draw[very thick, red] (6.75,-4.5)--  node [below] {$P_2$} (10,-4.5);
\end{tikzpicture}
\caption{The steps performed by our algorithm on Example~\ref{ex:unequal}. The colored intervals are the pieces already allocated to a player.}
\label{fig:unequal}
\end{figure}
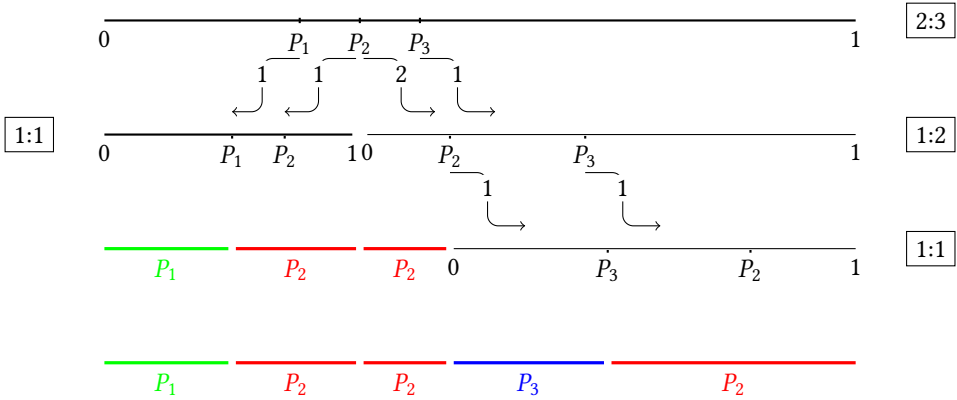

\begin{theorem}
Our ``Protocol for proportional division with unequal shares'' terminates with a proportional division.
\end{theorem}

\begin{proof}
We provide detailed calculations for the first subinstance only, because analogous calculations can easily be obtained for the second subinstance. First we observe that $d_j-a+\lfloor D/2 \rfloor = \lfloor D/2 \rfloor - \sum_{1}^{j-1}d_i$ is positive by the definition of $j$. Now we have to ensure that the subinstance is generated in such a manner that all players evaluate the full cake $X_1$ of the first subinstance equally and to the sum of all their demands. In the case of the first subinstance, the sum of demands is $\sum_{i=1}^{j}{d_i}-a+\lfloor \frac{D}{2}\rfloor = \lfloor \frac{D}{2}\rfloor$. This will be the measure of the cake $X_1$ for all players. To achieve this, $\mu_1, \mu_2, \dots, \mu_j$ need to be adjusted. Each $\mu_i$ will become in this subinstance
$$\mu_{i1} = \mu_i \cdot \frac{\lfloor\frac{D}{2}\rfloor}{\mu_i(X_1)}.$$
If $P_i$, $1 \leq i <j$ receives a piece of worth $d_i$, then in the original instance, it is of worth 
$$d_i \cdot \frac{\mu_i(X_1)}{\lfloor \frac{D}{2}\rfloor} \geq d_i,$$
because $\mu_i(X_1) \geq \mu_j(X_1) = \lfloor \frac{D}{2}\rfloor$, due to the cutting rule in our protocol. With this we have shown that every player appearing only in the first subinstance is guaranteed to gain her proportional share. An analogous proof works for players $P_{j+1}, P_{j+2}, \dots, P_{n}$. The last step is to show that $P_j$ collects her proportional share from the two subinstances. 

The only player whose measure function certainly need not be adjusted is~$P_j$. It is because $\mu_j(X_1) = \lfloor\frac{D}{2}\rfloor$, thus $\mu_{j1} = \mu_j \cdot \frac{\lfloor\frac{D}{2}\rfloor}{\mu_j(X_1)} = \mu_j$. Therefore, if $P_j$ receives her proportional share $d_j - a + \lfloor\frac{D}{2}\rfloor$ and $a - \lfloor\frac{D}{2}\rfloor$ in the two subinstances, then in the original instance her piece is worth $d_j$ at least.
\end{proof}

Having shown its correctness, we now present our estimation for the number of queries our protocol needs.

\begin{theorem}
\label{th:unequal}
For any $2\leq n$ and $n < D$, the number of queries in our $n$-player protocol on a cake of total value $D$ is
$T(n,D)\leq 2(n-1) \cdot \lceil\log_2{D}\rceil$.
\end{theorem}

\begin{proof}
If $n=2$, then our algorithm simulates the Cut Near-Halves algorithm---except that it uses cut queries exclusively---and according to Theorem~\ref{th:unequal_old} it requires $2\lceil\log_2{D}\rceil$ queries at most. This matches the formula stated in Theorem~\ref{th:unequal}. From this we prove by induction. For $n > 2$, the following recursion formula corresponds to our rules.
$$T(n,D) = n + \max_{1 \leq i\leq n}\left\{T(i, \lfloor\frac{D}{2}\rfloor) + T(n-i+1,\lceil\frac{D}{2}\rceil)\right\}$$
We now substitute our formula into the right side of this expression. 
\begin{equation*}
\begin{split}
n + \max_{1 \leq i\leq n}\left\{T(i,\lfloor\frac{D}{2}\rfloor) + T(n-i+1,\lceil\frac{D}{2}\rceil)\right\}&=\\
n + \max_{1 \leq i\leq n}\left\{2(i-1)\lceil\log_2{\lfloor\frac{D}{2}\rfloor}\rceil + 2(n-i)\lceil\log_2{\lceil\frac{D}{2}\rceil}\rceil\right\}&\leq  (*)\\
n+\max_{1 \leq i\leq n}\left\{2(i-1)(\lceil\log_2{D}\rceil-1) + 2(n-i)(\lceil \log_2{D}\rceil-1)\right\}&=\\
n + 2(n-1)(\lceil\log_2{D}\rceil-1)&=\\
-n + 2 + 2(n-1)\lceil\log_2{D}\rceil&\leq\\
 2(n-1) \cdot \lceil\log_2{D}\rceil&=T(n,D) \end{split}
\end{equation*}
The inequality marked by $(*)$ is trivially correct if $D$ is even. For odd $D$, we rely on the fact that $\log_2{D}$ cannot be an integer.

$$\lceil\log_2{\lfloor\frac{D}{2}\rfloor}\rceil \leq \lceil\log_2{\lceil\frac{D}{2}\rceil}\rceil = \lceil\log_2{\frac{D+1}{2}}\rceil = \lceil\log_2{(D+1)} - \log_2{2}\rceil = \lceil\log_2{(D+1)}\rceil - 1 = \lceil\log_2{D}\rceil - 1 $$
\end{proof}

With a query number of $\mathcal{O}(n \log_2{D})$, our protocol is more efficient than all known protocols. 
We will now point out a further essential difference in fairness when comparing to the fastest known protocol before our result, the generalized Cut Near-Halves. Our protocol treats players equally, while the generalized Cut Near-Halves does not. Equal treatment of players is a clear advantage if one considers the perception of fairness from the point of view of a player. 

We remark that our protocol is not truthful, which can be illustrated on a simple example. Take the 2-player equal shares case with nonzero measure functions on any nonzero measure interval. If the player whose mark is at the left knows the measure function of the other player, she can easily manipulate the outcome by marking the half of the cake just before the mark of the other player. As a result, her piece will be larger than what she receives if she reports the truth, unless their measure functions are special. 

\begin{remark} In the ``Protocol for proportional division with unequal shares''
\begin{itemize}
\item each player answers the exact same queries as the other players in the same round and same subinstance;
\item no player is asked to disclose the outcome of an eval query.
\end{itemize}
\end{remark}

\begin{proof}
In any subinstance, our protocol asks each player to answer the same proportional cut query, namely cutting the current cake to near-halves. Eval queries in these proportional queries are only utilized as technical workarounds to determine the value of the piece that plays the cake in the current subinstance. Their result is never revealed to any other player or even the protocol itself. The only outcome of the proportional cut query is a mark at the near-half of the current cake. 
Moreover, there is no difference in the role of the players when queries are asked, and no player is doomed to receive her exact share, like the cutter in Cut-and-Choose. If we consider Cut-Near-Halves, being the cutter in the first round is the most undesired role, followed by being a cutter in the second round, and so on. Our protocol forgoes this differentiation between the players, since it addresses the same queries to each player in a round, and the cake will be cut at the mark of the player whose demand happens to reach $\lfloor \frac{D}{2}\rfloor$ when the demands are summed up in order of the marks on the cake.

\end{proof}

The generalized Cut Near-Halves protocol fails to satisfy both of the above points. It addresses both eval and cut queries to players and treats players differently based on which type of query they got. In the 2-player version of Cut Near-Halves, only one player marks the cake and the other player uses an eval query to choose a side. This enables the second player to have a chance for a piece strictly better than half of the cake, while the first player is only entitled for her exact proportional share and has no chance to receive more than that. Besides this, the player who is asked to evaluate a piece might easily speculate that she was offered the piece because the other player cut it off the cake.

However, the remark is true for the Even-Paz protocol for proportional division with equal shares, which can be utilized in our problem through the cloning technique. As mentioned in Section~\ref{sec:known}, it needs $\mathcal{O}(D \log_2{D})$ proportional cut queries. The more efficient generalized Cut Near-Halves protocol only needs $\mathcal{O}(n^2 \log_2{D})$ queries, but it treats players differently. Our protocol adheres to the equal treatment of players principle and beats both protocols in efficiency.
\section{Generalizations}
\label{sec:gen}

In this section we introduce a far generalization of cake cutting, where the cake is a measurable set in arbitrary finite dimension and cuts are defined by a monotone function. 
At the end of the section we prove that even in the generalized setting, $\mathcal{O}(n \log_2{D})$ queries suffice to construct a proportional division.

\subsection{A general cake definition}
\label{sec:generalcake}

Our players remain $\{P_1, P_2, \dots, P_n\}$ with demands $d_i \in \mathbb{Z}^+$, but the cake is now a Lebesgue-measurable subset $X$ of $\mathbb{R}^k$ such that $0< \lambda(X) < \infty$. Each player $P_i$ has a non-negative, absolutely continuous \emph{measure function} $\mu_i$ defined on the Lebesgue-measurable subsets of~$X$. An important consequence of this property is that for every $Z \subseteq X$, $\mu_i(Z)=0$ if and only if $\lambda(Z)=0$. The value of the whole cake is identical for all players, in particular it is the sum of all demands:
\begin{equation*}
\forall 1 \leq i \leq n \quad \mu_i(X)= D = \sum\limits_{j=1}^n d_j.
\end{equation*}
A measurable subset $Y$ of the cake $X$ is called a \emph{piece}. The \textit{volume} of a piece $Y$ is the value $\lambda(Y)$ taken by the Lebesgue-measure on~$Y$. 
The cake $X$ will be partitioned into pieces $X_1,\ldots,X_n$.

\begin{definition}
A set $\lbrace X_i \rbrace_{1 \leq i \leq n}$ of pieces is a \emph{division} of $X$ if $\bigcup\limits_{1 \leq i \leq n}X_i =X$ and $X_i \cap X_j = \emptyset$ holds for all $i \neq j$. We call division $\lbrace X_i \rbrace_{1 \leq i \leq n}$ \emph{proportional} if $\mu_i(X_i)\geq d_i$ holds for all $1 \leq i \leq n$.
\end{definition}

We will show in Section~\ref{sec:gen_prot} that a proportional division always exists.

\subsection{A stronger query definition}
\label{sec:gen_query}
The more general cake clearly requires a more powerful query notion. Cut and eval queries are defined on an arbitrary piece (i.e. measurable subset) $I \subseteq X$. Beyond this, each cut query specifies a value $\alpha \in \mathbb{R}^+$ and a monotone mapping $f:[0,\lambda(I)]\to 2^I$ (representing a moving knife) such that $f(x)\subseteq f(y)$ and 
$\lambda(f(x))=x$ holds for every $0\le x\le y\le \lambda(I)$. 
\begin{itemize}
\item \textit{Eval query} $(P_i,I)$ returns $\mu_i(I)$.
\item \textit{Cut query} $(P_i, I, f,\alpha)$ returns an
$x\le\lambda(I)$ with $\mu_i(f(x))=\alpha$ or an error message if such an $x$ does not exist.
\end{itemize}
As queries involve an arbitrary measurable subset $I$ of $X$, our generalized queries automatically cover the generalization of the previously discussed Edmonds-Pruhs queries, proportional queries and reindexing. If we restrict our attention to the usual unit interval cake $[0,1]$, generalized queries open up a number of new possibilities for a query, as Example~\ref{ex:gen_interval} shows. 

\begin{example}
\label{ex:gen_interval}
On the unit interval cake the following rules qualify as generalized queries.
\begin{itemize}
\item Evaluate an arbitrary measurable set.
\item Cut a piece of value $\alpha$ surrounding a point $x$ so that $x$ is the midpoint of the cut piece.
\item For disjoint finite sets $A$ and $B$, cut a piece $Z$ of value $\alpha$ such that $Z$ contains the $\varepsilon$-neighborhood of $A$ and avoids the $\varepsilon$-neighborhood of $B$ for a maximum $\varepsilon$.
\item Determine $x$ such that the union of intervals $[0,x], [\frac 1n, \frac 1n+x],\ldots, [\frac{n-1}n, \frac {n-1}n+x]$ is of value~$\alpha$.
\end{itemize}
\end{example}

\noindent The new notions also allow us to define cuts on a cake in higher dimensions. 

\begin{example}
Defined on the generalized cake $X \subseteq \mathbb{R}^k$, the following rules qualify as generalized queries.
\begin{itemize}
\item Evaluate an arbitrary measurable set.
\item Cut a piece of value $\alpha$ of piece $I$ so that the cut is parallel to a given hyperplane.
\item Multiple cut queries on the same piece $I \subset \mathbb{R}^2$: one player always cuts $I$ along a horizontal line, the other player cuts the same piece along a vertical line.
\end{itemize}
\end{example}

\subsection{The existence of a proportional division}
\label{sec:gen_prot}

Our algorithm ``Proportional division with unequal shares'' in Section~\ref{sec:protocol} extends to the above described general setting and hence proves that a proportional division always exists.

\begin{theorem}
\label{th:gen_unequal}
For any $2\leq n$ and $n < D$, the number of generalized queries in our $n$-player protocol on the generalized cake of total value $D$ is
$T(n,D)\leq 2(n-1) \cdot \lceil\log_2{D}\rceil$.
\end{theorem}

\begin{proof}
The proof of Theorem~\ref{th:unequal_old} carries over without essential changes, thus we only discuss the differences here. First we observe that proportional queries in ratio $a:b$ can still be substituted by a constant number of eval and cut queries. In the generalized model, proportional query $(P_i, I, f, a, b)$ returns $x\le\lambda(I)$ such that $b\cdot \mu_i(f(x))=a\cdot \mu_i(I\setminus f(x))$. Similarly as before, $P_i$ first measures $I$ by a single eval query and then uses the cut query  $(P_i, I, f,\alpha)$ with $\alpha = \frac{a}{a+b} \cdot \mu_i(I)$. In the first round of our generalized algorithm, all players are asked to cut the cake $X$ in near-halves using the same $f$ function. Then $P_j$ is calculated, just as in the simpler version and we cut $X$ into the two near-halves according to $P_j$'s $f$-cut and clone $P_j$ if necessary. Due to the monotonicity of $f$, this sorts each player to a piece she values at least as much as the full demand on all players sorted to that piece. Subsequent rounds are played in the same manner.

The query number for $n=2$ follows from the fact that each of the two players are asked a proportional cut query in every round until recursively halving $\lceil\frac{D}{2}\rceil$ reaches 1, which means $\lceil\log_2{D}\rceil$ queries in total. The recursion formula remains intact in the generalized model, and thus the query number $T(n,D) = 2 (n-1) \lceil\log_2{D}\rceil$ too.
\end{proof}

\section{The lower bound}
\label{sec:lower_bound}

In this section, we prove our lower bound on the number of queries any deterministic protocol needs to make when solving the proportional cake cutting problem with unequal shares. This result is valid in two relevant settings: \begin{inparaenum} \item on the $[0,1]$ cake with Robertson-Webb or with generalized queries, \item on the general cake and queries introduced in Section~\ref{sec:gen}.\end{inparaenum} 

The lower bound proof is presented in two steps. In Section~\ref{sec:single_player} we define a single-player cake-cutting problem where the goal is to identify a piece of small volume and positive value for the sole player. For this problem, we design an adversary strategy and specify the minimum volume of the identified piece as a function of the number of queries asked. In Section~\ref{sec:n_player} we turn to the problem of proportional cake cutting with unequal shares. We show that in order to allocate each player a piece of positive value, at least $\Omega(n \log{D})$ queries must be addressed to the players---otherwise the allocated pieces overlap.

\subsection{The single-player problem}
\label{sec:single_player} 
We define our single-player problem on a generalized cake of value $D$, a player $P$ and her unknown measure function~$\mu$. The aim is to identify a piece of positive value according to~$\mu$ by asking queries from~$P$. The answers to these queries come from an adversary strategy we design. We would like to point out that the single-player \textit{thin-rich game} of \citet{EP11} defined on the unit interval cake has a different goal. There, the player needs to receive a piece that has value not less than $1$ and width at most~$2$. Moreover, their proof for the $n$-player problem is restricted to instances with $n = 2\cdot 3^{\ell}, \ell \in \mathbb{Z}^+$, whereas ours is valid for any $n\in \mathbb{Z}^+$.

In our single-player problem, a set of queries reveals information on the value of some pieces of the cake. Each generalized eval query $(P, I)$ partitions the cake into two pieces; $I$ and $X \setminus I$. An executed cut query $(P, I, f, \alpha)$ with output $x$ partitions the cake into three; $f(x)$, $I\setminus f(x)$ and $X\setminus I$. To each step of a protocol we define the currently smallest building blocks of the cake, which we call \textit{crumbles}. Two points of $X$ belong to the same crumble if and only if they are in the same partition in all queries asked so far. At start, the only crumble is the cake itself and every new query can break an existing crumble into more crumbles. More precisely, $q$ queries can generate $3^q$ crumbles at most. Crumbles at any stage of the protocol partition the entire cake. The exact value of a crumble is not necessarily known to the protocol and no real subset of a crumble can have a revealed value. As a matter of fact, the exact same information are known about the value of any subset of a crumble.

\begin{example}
\label{ex:crumble}
In Figure~\ref{fig:crumble} we illustrate an example for crumbles on the unit interval cake after two queries. The upper picture depicts a cut query defined on the green set~$I$. It generates a piece of value $\alpha$ so that it contains the $\varepsilon$-neighborhood of points $A_1, A_2, A_3$ for maximum~$\varepsilon$. This piece is marked red in the figure and it is a crumble. The second crumble  at this point is the remainder of $I$ (marked in green only), while the third crumble is the set of points in black. These three crumbles are illustrated in the second picture. The second query evaluates the blue piece in the third picture. It cuts the existing crumbles into 6 crumbles in total, as depicted in the bottom picture.
\end{example}

\begin{figure}[t]
\centering
\begin{tikzpicture}[every edge/.style={shorten <=1pt, shorten >=1pt}]

  \draw[thick] (0,0)  node [below] {0} -- (10,0) node [below] {1};
    \foreach \y/\ytext in {1.1/$I_1$,8.0/$I_2$}
    \draw[thick] (\y,1pt) -- (\y,-1pt);
    \draw[line width=2mm, green] (1.1,0) -- (8,0);
  \foreach \x/\xtext in {2.8/$A_1$,3.4/$A_2$,5.6/$A_3$}
    {\draw[thick] (\x,1pt) -- (\x,-1pt) node[below] {\xtext};
    \draw[ultra thick, red] (\x-1,0) -- (\x+1,0);}
\node[draw] at (11,0) {cut};

  \draw[ultra thick] (0,-1.5)  node [below] {0} -- (10,-1.5) node [below] {1};
  \foreach \x/\xtext in {2.8/$A_1$,3.4/$A_2$,5.6/$A_3$}
    {\draw[thick] (\x,1pt) -- (\x,-1pt) node[below] {};
    \draw[ultra thick, red] (\x-1,-1.5) -- (\x+1,-1.5);}
    \draw[ultra thick, green] (1.1,-1.5) -- (1.8,-1.5);
    \draw[ultra thick, green] (4.4,-1.5) -- (4.6,-1.5);
    \draw[ultra thick, green] (6.6,-1.5) -- (8,-1.5);
\node[] at (11.5,-1.5) {3 crumbles};

	\draw[thick] (0,-3)  node [below] {0} -- (10,-3) node [below] {1};
  \foreach \x/\xtext in {1.8/$A_1$,6.7/$A_2$}
    \draw[ultra thick, blue] (\x-1,-3) -- (\x+1,-3);
\node[draw] at (11,-3) {eval};

\draw[ultra thick] (0,-4.5)  node [below] {0} -- (10,-4.5) node [below] {1};

	\draw[ultra thick, black] (0,-4.5)--  (0.8,-4.5);
    \draw[ultra thick, black] (8,-4.5)--  (10,-4.5);
    \draw[ultra thick, blue] (0.8,-4.5)--  (1.1,-4.5);
    \draw[ultra thick, green] (7.7,-4.5)--  (8,-4.5);
    \draw[ultra thick, orange] (1.1,-4.5)--  (1.8,-4.5);
    \draw[ultra thick, orange] (6.6,-4.5)--  (7.7,-4.5);
    \draw[ultra thick, yellow] (1.6,-4.5)--  (2.8,-4.5);
    \draw[ultra thick, yellow] (5.7,-4.5)--  (6.6,-4.5);
     \draw[ultra thick, red] (2.8,-4.5)--  (4.4,-4.5);
      \draw[ultra thick, red] (4.6,-4.5)--  (5.7,-4.5);
      \draw[ultra thick, green] (4.4,-4.5)--  (4.6,-4.5);
    
\node[] at (11.5,-4.5) {6 crumbles};
\end{tikzpicture}
\caption{The crumble partition after two queries in Example~\ref{ex:crumble}. We marked each of the 6 crumbles by a different color in the bottom picture.}
\label{fig:crumble}
\end{figure}
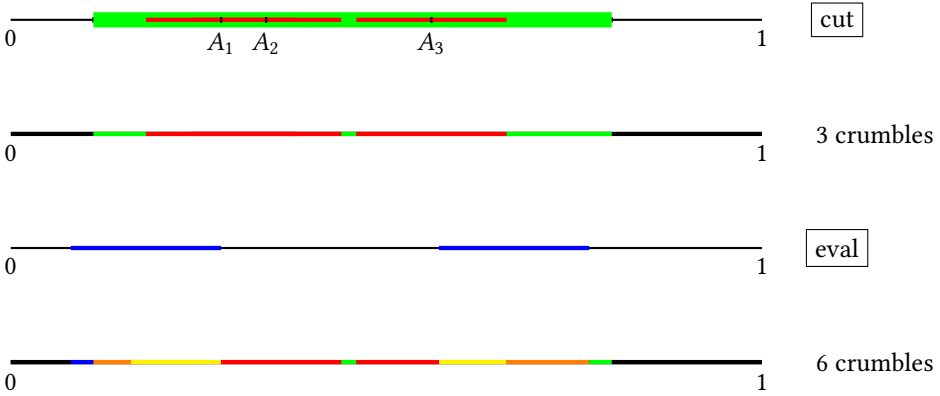

We now proceed to construct an adversary strategy that bounds the volume of any crumble $C$ with $\mu(C) > 0$. Our adversary can actually reveal more information than asked; we allow her to disclose the value of each crumble in the cake. When a query is asked, the answer is determined based on the parameters of the query and the current set of crumbles, which we denote by~$\mathcal{C}$. Together with the answer to the query, the adversary also specifies the new set of crumbles $\mathcal{C}^{new}$ together with $\mu(C^{new})$ for each $C^{new} \in \mathcal{C}^{new}$. In the next query, this $\mathcal{C}^{new}$will serve as the current set of crumbles $\mathcal{C}$. The adversary answers the queries in accordance to the following rules, which are also stated in a pseudocode below. 

\begin{itemize}
	\item eval query $(\mathcal{C}, I)$ \\
		This query changes the structure of the crumble set $\mathcal{C}$ in such a way that each crumble $C \in \mathcal{C}$ is split into exactly two new crumbles $C\cap I$ and $C\setminus I$, both of which might be empty (lines~1-2). If the part inside the crumble is at least as large as the other part, then the adversary assigns the full vale of $C$ to $C \cap I$ (lines~3-5). Otherwise, the outer part $C \setminus I$ will get the entire value (lines~6-8). The answer to the eval query is the total value of new crumbles that lie in $I$ (line~11).
	\item cut query $(\mathcal{C}, I, f, \alpha)$\\  
Each cut query is executed in two rounds. In the first round (lines~12-22) we define new crumbles $C\setminus I$ (line~13) and \emph{intermediate crumbles} $C\cap I$ (line~14) for all crumbles $C \in \mathcal{C}$. If $\lambda(C\cap I) \ge 2/3\cdot \lambda (C)$ then $C\cap I$ inherits the entire value of $C$ (lines~15-17), otherwise $C\setminus I$ carries all the value of $C$ (lines~18-20). The new crumbles  are set aside until the next query arrives, while the intermediate crumbles will be the crumbles of the second round (lines~23-45).

If the total value of these intermediate crumbles is less than $\alpha$, then an error message is returned indicating that $I$ is not large enough to be cut off a piece of value $\alpha$ (lines~23-24). Otherwise, for each intermediate crumble $C^{int}_i$ we define the value $x_i$ for which $f(x_i)$ halves $C^{int}_i$ in volume (lines~26-28). We then reorder the indices of intermediate crumbles according to these $x_i$ values (line~29). Now we find the index $k$ for which $\sum_{i=1}^{k-1}\mu(C^{int}_i)<\alpha \le \sum_{i=1}^{k}\mu(C^{int}_i)$.

The set of new crumbles will now be completed by adding sets $\left\{C^{int}_i \cap f(x_k)\right\}$ and $\left\{C^{int}_i \setminus f(x_k)\right\}$ to it (lines~31-32). The value of these new crumbles is specified depending on the index $i$ of $C^{int}_i$. If $i < k$, then the crumble in $f(x_k)$ inherits the full value of the intermediate crumble (lines~34-36). If $i > k$, then the crumble outside of $f(x_k)$ inherits the value of the intermediate crumble (lines~37-39). Finally, for $i = k$, the crumble inside $f(x_k)$ receives all of $\alpha$ that has not been assigned to new crumbles inside $f(x_k)$ with a smaller index (line~41). After this, the crumble outside of $f(x_k)$ gets the remainder of $\mu(C^{int}_k)$ (line~42). At last, the algorithm returns $x=x_k$ (line~44). 
\end{itemize}

\begin{algorithm}[t]
\caption{Adversary strategy}
\begin{minipage}[t]{0.35\textwidth}
\textsc{Eval query} $(\mathcal{C},I)$\\\nl
\everypar={\nl}
\For{$\forall C \in \mathcal{C}$}{
		$\mathcal{C}_{new} \leftarrow \mathcal{C}_{new} \cup \left\{C \cap I\right\} \cup \left\{C \setminus I\right\} $\\\nl
  \eIf{$\lambda(C \cap I) \geq \frac{1}{2}\lambda(C)$}{
		$\mu(C \cap I) \leftarrow \mu(C)$\\\nl
		$\mu(C \setminus I) \leftarrow 0$
		}{
		$\mu(C \cap I) \leftarrow 0$\\\nl	
		$\mu(C \setminus I) \leftarrow \mu(C)$
  }
 }
return $\sum_{C_{new} \in \mathcal{C}_{new}, C_{new} \subseteq I}{\mu(C_{new})}$
\end{minipage}\hspace{7mm}
\begin{minipage}[t]{0.56\textwidth}
\textsc{Cut query} $(\mathcal{C},I,f,\alpha)$\\\nl
\everypar={\nl}
\For{$\forall C \in \mathcal{C}$}{
		$\mathcal{C}^{new} \leftarrow \mathcal{C}^{new} \cup \left\{C \setminus I\right\}$\\\nl
		$\mathcal{C}^{int} \leftarrow \mathcal{C}^{int} \cup \left\{C \cap I\right\}$\\\nl
  \eIf{$\lambda(C \cap I) \geq \frac{2}{3}\lambda(C)$}{
		$\mu(C \cap I) \leftarrow \mu(C)$\\\nl
		$\mu(C \setminus I) \leftarrow 0$
		}{
		$\mu(C \cap I) \leftarrow 0$\\\nl	
		$\mu(C \setminus I) \leftarrow \mu(C)$
  }
 }
\eIf{$\sum_{C^{int} \in \mathcal{C}^{int}}{\mu(C^{int})} < \alpha$}{
	return error
}
{
	\For{$\forall C^{int}_i \in \mathcal{C}^{int}$}{
	find $x_i \in \mathbb{R}$ so that $\lambda(C^{int} \cap f(x_i)) = \frac{1}{2}\lambda(C^{int}_i)$
	}
	reorder $\left[i\right]$ in $C^{int}_i$ so that $x_1 \leq x_2 \leq \ldots$\\\nl
	find $k \in \mathbb{Z}$ so that $\sum_{i=1}^{k-1}\mu(C^{int}_i)<\alpha \le \sum_{i=1}^{k}\mu(C^{int}_i)$ \\\nl
	\For{$\forall C^{int}_i \in \mathcal{C}^{int}$}{
		$\mathcal{C}^{new} \leftarrow \mathcal{C}^{new} \cup \left\{C^{int}_i \cap f(x_k)\right\}  \cup \left\{C^{int}_i \setminus f(x_k)\right\}$
	}
	\uIf {$ i < k$}{
	$\mu(C^{int}_i\cap f(x_k)) \leftarrow \mu(C^{int}_i)$ \\\nl
	$\mu(C^{int}_i\setminus f(x_k)) \leftarrow 0$
	}
	\uElseIf {$i>k$}{
	$\mu(C^{int}_i\cap f(x_k)) \leftarrow 0$ \\\nl
	$\mu(C^{int}_i\setminus f(x_k)) \leftarrow \mu(C^{int}_i)$
	}
	\Else{
	$\mu(C^{int}_k\cap f(x_k))\leftarrow\alpha -\sum _{i=1}^{k-1}\mu(C^{int}_i)$\\\nl
	$\mu(C^{int}_k\setminus f(x_k)) \leftarrow \mu(C^{int}_k)+\sum _{i=1}^{k-1}\mu(C^{int}_i)-\alpha$	
}
return $x_k$
}
\end{minipage}
\end{algorithm}

Once all queries have been answered according to the above rules, the player is allocated a piece~$Z \subseteq X$. The adversary specifies $\mu(Z)$ as the total value of those crumbles that are subsets of~$Z$.

Having described and demonstrated our adversary strategy, we now turn to proving our key lemma on the volume of pieces that carry a positive value.

\begin{lemma}
\label{le:2}
	After $q$ queries in the single-player problem, the volume of any piece with positive value is at least $\frac{D}{3^q}$.
\end{lemma}
\begin{proof}
Due to the last rule of the adversary strategy, the volume of any piece with positive value is bounded from below by the volume of any crumble with positive value. We will now argue that eval and cut queries assign positive value to crumbles whose volume is at least a third of the volume of the previous crumble.
 
At the very beginning of the protocol, for $q=0$, the only crumble is $X$ itself, with volume $\frac{D}{3^0}$. Eval queries assign positive value to crumbles that are at least as large as half of the previous crumble they belonged to prior to the query. 
If a new crumble with positive value was created in the first round of a cut query, then its volume was at least one third of a previous crumble (lines 13 and~20). Otherwise, the new crumble with positive value was an intermediate crumble $C^{int}_i$ in the second round. The first round of our algorithm assigns positive value to an intermediate crumble only if it was at least two-thirds of the old crumble in the input of the cut query (lines~14-15). This round will now cut $C^{int}_i$ into two new crumbles (line~32). If $i \neq k$, then the larger of these will inherit the value of the intermediate crumble (lines 35 and~39). Otherwise, if $i=k$, then $C^{int}_i$ is cut into exact halves (lines~41-42). All in all, new crumbles that are assigned a positive value  in the second round are of volume at least half of two-thirds of the volume of the original crumble in the input of the query.
\end{proof}




\subsection{The $n$-player problem}
\label{sec:n_player} 

We now place our single-player problem into the framework of the original problem. The instance we construct has $c_1n$ players whose demand sums up to $c_2n$, where $c_1$ and $c_2$ are arbitrary constants between $0$ and~$1$. We call these players \emph{humble}, because their total demand is modest compared to the number of them. The remaining $(1-c_1)n$ players share a piece of worth $D-c_2n$. These players are \emph{greedy}, because their total demand is large. The simplest such instance is where $n-1$ humble players have demand~1, and the only greedy player has demand $D-(n-1)$. We fix the measure function of every greedy player to be the Lebesgue-measure. This enforces humble players to share a piece of volume $c_2D$ among themselves. Lemma~\ref{le:2} guarantees that after $q_i$ queries addressed to $P_i$, the volume of any piece carrying positive value for $P_i$ is at least $\frac{D}{3^{q_i}}$. We now sum up the volume of the pieces allocated to humble players in any proportional division.

$$\sum_{i=1}^{c_1n}\frac{D}{3^{q_i}} \leq c_2n$$
We divide both sides by $c_1 nD$.
$$\frac{1}{c_1n}\sum_{i=1}^{c_1n}3^{-q_i} \leq \frac{c_2}{c_1D}\\$$
For the left side of this inequality, we use the well known inequality for the arithmetic and geometric means of non-negative numbers.
$$\sqrt[c_1n]{3^{-\sum_{i_1}^{c_1n}{q_i}}} \leq \frac{c_2}{c_1D}$$
Taking the logarithm of both sides leads to the following.
$$\frac{1}{c_1 n} \left(-\sum_{i_1}^{c_1n}{q_i}\right) \leq \log_3{\frac{c_2}{c_1}}- \log_3{D}$$
With this, we have arrived to a lower bound on the number of queries.
$$\sum_{i_1}^{c_1n}{q_i} \geq c_1n \left( \log_3{D}+\log_3{\frac{c_1}{c_2}}\right) \sim \Omega(n \log_3D)$$

This proves that one needs $\Omega(n \log_3D)$ queries to derive a proportional division for the humble players in the instance. Moreover, if $c_1$ and $c_2$ are known, a more accurate bound can be determined using our formula $c_1n \left( \log_3{D}-\log_3{c_2}+\log_3{c_1}\right)$. This suggests that the problem becomes harder to solve if the $c_1 n$ humble players vastly outnumber the greedy players. In the $c_1 n = n-1$ case we mentioned earlier, the query number is at least $(n-1) \log_3{D}$.



We can now conclude our theorem on the lower bound.

\begin{theorem}
\label{th:lower}
	To construct a proportional division in an $n$-player unequal shares cake cutting problem with demands summing up to $D$ one needs $\Omega(n \log{D})$ queries.
\end{theorem}


\section{Irrational demands}
\label{sec:irrat}
In this section we consider the case when some demands are irrational numbers. Apart from this, our setting is exactly the same as before. Even though two direct protocols have been presented for the problem of proportional cake cutting with irrational demands~\cite{Bar96,SZ99}, we feel that our protocol sheds new light to the topic. The complexity of all known protocols for irrational shares falls into the same category: finite but unbounded. \citet{SZ99} present a protocol that is claimed to be simpler than the one of \citet{Bar96}. First they present a 2-player protocol, in which one player marks a large number of possibly overlapping intervals that are worth the same for her. The other player then chooses one of these so that it satisfies her demand. The authors then refer to the usual inductive method to the case of $n$ players, in which the $n$-th player shares each of the $n-1$ pieces the other players have already obtained. This procedure is cumbersome compared to our protocol that reduces the problem to one with rational demands or decreases the number of players. Moreover, our method works on our generalized cake and query model.


Let us choose an arbitrary piece $A\subseteq X$ such that $\mu_i(A) > 0$ for all players $P_i$. If the players share $A$ and $X \setminus A$ in two separate instances, both in their original ratio $d_1 : d_2 : \ldots : d_n$, then the two proportional divisions will give a proportional $d_1 : d_2 : \ldots : d_n$ division of $X$ itself. Assume now that $\mu_i(A) < \mu_j(A)$ for some players $P_i$ and $P_j$, and some piece $A\subseteq X$. When generating the two subinstances on $A$ and $X \setminus A$, we reduce $d_i$ on $A$ to 0 and increase it in return on $X \setminus A$ and and swap the roles for $d_j$, increasing it on $A$ and decreasing it on $X \setminus A$. The first generated instance thus has $n-1$ players with irrational demands, while the second instance has $n$ players with irrational demands. We will show in Lemma~\ref{le:2a} that if we set the right new demands in these instances, the two proportional divisions deliver a proportional division of~$X$. The key point we prove in Lemma~\ref{le:2b}, which states that the demands in the second subinstance sum up to slightly below all players' evaluation of $X \setminus A$. Redistributing the slack as extra demand among players gives us the chance to round the demands up to rational numbers in the second subinstance and keep proportionality in the original instance. Iteratively breaking up the instances into an instance with fewer players and an instance with rational demands leads to a set of instances with rational demands only.  

We now describe our protocol in detail. Without loss of generality we can assume that $d_1 \leq d_2 \leq \ldots \leq d_n$. As a first step, $P_1$ answers the cut query with $x=d_1$ and $I = X$. We denote the piece in $f(d_1)$ by~$A$ and ask all players to evaluate~$A$. Let $P_j$ be one of the players whose evaluation is the highest. Notice that $\mu_j(A)\geq d_1$, because $\mu_1(A) = d_1$. We distinguish two cases from here.

\begin{enumerate}
\item If $\mu_j(A) = d_1$, then $\mu_i(A) \leq d_1$ for all players. We allocate $A$ to $P_1$ and continue with an instance $\mathcal{I}_1$ with $n-1$ players having the same demands as before. The measure functions need to be normalized to $\frac{D-d_1}{D-\mu_i(A)} \cdot \mu_i$ for all $i \neq 1$ so that all players of $\mathcal{I}_1$ evaluate $X \setminus A$ to $D-d_1$.

\item Otherwise, $\mu_j(A) = d_1 + \varepsilon$, where $\varepsilon > 0$. We generate instances $\mathcal{I}_{2a}$ and~$\mathcal{I}_{2b}$.
\begin{enumerate}[(a)]
\item In the first instance $\mathcal{I}_{2a}$, the cake is $A$, $P_1$'s demand is 0, $P_j$'s demand is $d_j+d_1$, while all other players keep their original $d_i$ demand. In order to make all players evaluate the full cake to the sum of their demands $D$, measure functions are modified to $\frac{D}{\mu_i(A)}\cdot\mu_i$.

\item In the second instance $\mathcal{I}_{2b}$, the cake is $X \setminus A$, $P_1$'s demand is $d_1 + \frac{d_1^2}{D-d_1}$, $P_j$'s demand is $d_j-\frac{d_1 (d_1 + \varepsilon)}{D-(d_1 + \varepsilon)}$, while the original $d_i$ demands are kept for all other players. In order to make all players evaluate the full cake to $D$, we set $\frac{D}{D-\mu_i(A)}\cdot \mu_i$.

\end{enumerate}
\end{enumerate}

\begin{center}
\hspace*{0.2cm}\fbox{
\hspace*{2.5mm}\hspace*{-1\fboxsep}
\parbox{0.9\linewidth}{
\begin{minipage}[tb]{0.95\linewidth}
\begin{description}
    \item[\textbf{Proportional division with irrational demands}]
\end{description}
$P_1$ marks $d_1$ $\rightarrow A$. All players evaluate $A$. $P_j$ has the highest evaluation.\\

\begin{minipage}[l]{0.48\linewidth}
If $\mu_j(A) = d_1$, then allocate $A$ to $P_1$ and continue with $n-1$ players on $\mathcal{I}_1$.\\\\
Otherwise $\mu_j(A) = d_1 + \varepsilon$. Define two instances $\mathcal{I}_{2a}$ and $\mathcal{I}_{2b}$. While $\mathcal{I}_{2a}$ has $n-1$ players, demands in $\mathcal{I}_{2b}$ sum up to below $D$ and thus can be rationalized.
\end{minipage}\hspace{4mm}
\begin{minipage}[r]{0.3\linewidth}
		\begin{tabular}[\linewidth]{|l||c|c|c|}
		\hline
			&$\mathcal{I}_1$	& $\mathcal{I}_{2a}$ & $\mathcal{I}_{2b}$\\
			\hline
				cake & $X \setminus A$ & $A$ & $X \setminus A$\\
			\hline
			$d_1$ & 0 & 0 & $d_1 + \frac{d_1^2}{D-d_1}$\\
			\hline
			$d_j$ & $d_j$ & $d_j+d_1$ & $d_j-\frac{d_1 (d_1 + \varepsilon)}{D-(d_1 + \varepsilon)}$\\
			\hline
			$d_i$ & $d_i$ & $d_i$ & $d_i$ \\
			\hline
			$\mu_i$ & $\frac{D-d_1}{D-\mu_i(A)} \mu_i$ & $\frac{D}{\mu_i(A)}\mu_i$ & $\frac{D}{D-\mu_i(A)}\mu_i$\\
			\hline
		\end{tabular}
\end{minipage}
\end{minipage}
}}
\end{center} 
\vspace*{0.1cm}

\begin{lemma}
\label{le:1}
A proportional division in $\mathcal{I}_1$ extends to a proportional division in the original problem once $P_1$'s allocated piece $A$ is added to it.
\end{lemma}
\begin{proof}
Clearly $P_1$ is satisfied with $A$, since $\mu_1(A) = d_1$. In any proportional division in $\mathcal{I}_1$, every player $P_i$, $i \neq 1$ is guaranteed to receive a piece that is worth at least $\frac{D-\mu_i(A)}{D-d_1}\cdot d_i \geq d_i$ for her in the original instance.
\end{proof}

\begin{lemma}
\label{le:2a}
If each player receives her demanded share in $\mathcal{I}_{2a}$ and $\mathcal{I}_{2b}$, then the union of these pieces gives a proportional division in the original problem.
\end{lemma}
\begin{proof}
We calculate the share of each player for the case when each player receives a piece satisfying her demand in $\mathcal{I}_{2a}$ and~$\mathcal{I}_{2b}$.

\begin{itemize}
\item $d_1$: 0 + $(d_1 + \frac{d_1^2}{D-d_1}) \cdot \frac{D-d_1}{D} = d_1$
\item $d_j$: $(d_j + d_1) \cdot \frac{d_1 + \varepsilon}{D} + (d_j-\frac{d_1 \cdot (d_1 + \varepsilon)}{D-(d_1 + \varepsilon)}) \cdot \frac{D-(d_1 + \varepsilon)}{D}  = d_j$
\item $d_i$, $i \notin \{1,j\}$: $d_i \cdot \frac{d_1 + \varepsilon}{D} + d_i\cdot \frac{D-(d_1 + \varepsilon)}{D}  = d_i$
\end{itemize}
\end{proof}

\begin{lemma}
\label{le:2b}
By slightly increasing all demands, $\mathcal{I}_{2b}$ can be transformed into an instance of proportional cake cutting with rational demands.
\end{lemma}
\begin{proof}

The key observation here is that there is a slack in the demands, meaning that demands in $\mathcal{I}_{2b}$ sum up to strictly below $D$, which is the evaluation of all players of the full cake $X \setminus A$. The sum of the demands is the following.

$$d_1 + \frac{d_1^2}{D-d_1} + d_j-\frac{d_1 \cdot (d_1 + \varepsilon)}{D-(d_1 + \varepsilon)} + \sum_{i \notin \{1,j\}}{d_i} = \sum_{i=1}^n{d_i} + \frac{d_1^2}{D-d_1} -\frac{d_1 \cdot (d_1 + \varepsilon)}{D-(d_1 + \varepsilon)} < D$$

The inequality above follows from the fact that $\frac{d_1^2}{D-d_1} < \frac{d_1 (d_1 + \varepsilon)}{D-(d_1 + \varepsilon)}$ for all $\varepsilon >0$.

The slack can be distributed as extra demand among all players so that all demands are rational. An implementation of this could be that we round up the irrational demands at a sufficiently insignificant digit.
\end{proof}

\begin{theorem}
Any instance of the proportional cake cutting problem with $n$ players and irrational demands can be transformed into at most $n-1$ proportional cake cutting problems with rational demands and thus can be solved using a finite number of queries.
\end{theorem}

\begin{proof}
We prove this theorem by induction. 
For $n=2$, we need to show that the problem with irrational demands can be reduced to at most one proportional cake cutting problem with rational demands. For two players, our protocol proceeds as follows. First $P_1$ marks a piece $X$ that is worth $d_1$ for her. Now we ask $P_2$ to evaluate~$X$. If $\mu_2(A) \leq d_1$, then $P_1$ is allocated $X$ and $P_2$ is satisfied with $I \setminus X$, because $\mu_2(I \setminus X) \geq d_2$. Otherwise, if $\mu_2(X) = d_1 + \varepsilon$ for some $\varepsilon > 0$, then $P_2$ is allocated $X$ and the two players share $I \setminus X$ with demands $d_1 + \frac{d_1^2}{d_2}$ and $d_2 - \frac{d_1(d_1+\varepsilon)}{d_2-\varepsilon}$. These demands ensure a proportional share to both players, as we show in Lemma~\ref{le:2a}. Moreover, applying Lemma~\ref{le:2b} to two players proves that they sum up to strictly below $d_1+d_2$ and thus can be rounded up to rational numbers, which gives us the single 2-player problem that must be solved in order to derive a proportional division for the original problem. 

Assume now that an $n-1$-player proportional cake cutting problem with irrational demands can be solved by transforming it into $n-2$ proportional cake cutting problems with rational demands. If we are given an $n$-player proportional cake cutting problem with irrational demands, our ``Proportional division with irrational demands'' protocol transforms it into either $\mathcal{I}_1$ or to a problem with two instances, $\mathcal{I}_{2a}$ and $\mathcal{I}_{2b}$. Solving either of those problems will lead to a proportional division in the original $n$-player problem, as Lemmas~\ref{le:1} and~\ref{le:2a} show. The first instance is an $n-1$ player proportional cake cutting problem with irrational demands, which can be solved via $n-2$ proportional cake cutting problems with rational demands by our assumption. The same is true for $\mathcal{I}_{2a}$. As of $\mathcal{I}_{2b}$, Lemma~\ref{le:2b} proves that its demands can be rounded up to rational numbers. Even in the worst case, when our protocol generates $\mathcal{I}_{2a}$ and $\mathcal{I}_{2b}$ in every recursive step, it ends up constructing $n-1$ instances with rational demands.
\end{proof}

We would like to emphasize that even though we have transformed any proportional cake cutting problem with irrational demands into a set of problems with rational demands, we did not show any upper bound on its query complexity. When the problems with rational demands are created, $D$ might grow arbitrarily large, which hugely affects the query number.

\begin{acks}
We thank Simina Br\^{a}nzei and Erel Segal-Halevi for their generous and insightful advice. 
\end{acks}

\bibliographystyle{ACM-Reference-Format}
\bibliography{mybib}

\end{document}